\documentclass[11pt]{article}

\usepackage{fancyhdr}
\usepackage{amsmath}
\usepackage{amsthm}
\usepackage{amssymb,amsfonts}
\usepackage{threeparttable}
\usepackage[mathscr]{eucal}
\usepackage{graphicx}
\usepackage{geometry}
\usepackage{times}
\usepackage{enumerate}
\usepackage{courier}
\usepackage[round]{natbib}
\usepackage{float}
\usepackage{mathtools}
\usepackage{color}
\usepackage{authblk}

\newcommand{\bX}{\mathbf{X}} 
\newcommand{\by}{\mathbf{y}} 
\newcommand{\bz}{\mathbf{z}} 
\newcommand{\bb}{\mathbf{b}} 
\newcommand{\bbeta}{{\boldsymbol{\beta}}}
\newcommand{\calH}{\mathcal{H}} 
\newcommand{\calK}{\mathcal{K}} 
\newcommand{\calN}{\mathcal{N}} 
\newcommand{\calF}{\mathcal{F}} 
\newcommand{\RR}{\mathbb{R}} 
\newcommand{\EE}{\mathbb{E}}
\newcommand{\PP}{\mathbb{P}}
\newcommand{\tFDP}{\mathrm{FDP}}
\newcommand{\tTPP}{\mathrm{TPP}}
\newcommand{\tFDR}{\mathrm{FDR}}
\newcommand{\fdpinfty}{\mathrm{FDP}^{\infty}}
\newcommand{\tppinfty}{\mathrm{TPP}^{\infty}}
\newcommand{\fdpinftyprime}{\mathrm{FDP}^{ \infty}_{\text{aug}}}
\newcommand{\tppinftyprime}{\mathrm{TPP}^{ \infty}_{\text{aug}}}
\newcommand{\fdphatinftyprime}{\widehat{\mathrm{FDP}}^{ \infty}_{\text{aug}}}
\newcommand{\Xtil}{\widetilde{\mathbf{X}}}

\newtheorem{proposition}{Proposition}[section]

\newtheorem{lemma}[proposition]{Lemma}

\newtheorem{theorem}[proposition]{Theorem}

\DeclareMathOperator*{\argmin}{argmin}
\DeclarePairedDelimiter\ceil{\lceil}{\rceil}

\begin{document}

\title{A Power and Prediction Analysis\\ for Knockoffs with Lasso Statistics}
\author[a]{Asaf Weinstein}
\author[b]{Rina Barber}
\author[a]{Emmanuel Cand\`{e}s}
\affil[a]{Stanford University}
\affil[b]{University of Chicago}
\date{}                     

\maketitle
\thispagestyle{empty}

\begin{abstract}
  Knockoffs is a new framework for controlling the false discovery
  rate (FDR) in multiple hypothesis testing problems involving complex
  statistical models.  While there has been great emphasis on Type-I
  error control, Type-II errors have been far less studied.  In this
  paper we analyze the false negative rate or, equivalently, the power
  of a knockoff procedure associated with the Lasso solution path
  under an i.i.d.~Gaussian design, and find that knockoffs
  asymptotically achieve close to optimal power with respect to an
  omniscient oracle.  Furthermore, we demonstrate that for sparse
  signals, performing model selection via knockoff filtering achieves
  nearly ideal prediction errors as compared to a Lasso oracle
  equipped with full knowledge of the distribution of the unknown
  regression coefficients.  The i.i.d.~Gaussian design is adopted to
  leverage results concerning the empirical distribution of the Lasso
  estimates, which makes power calculation possible for both knockoff
  and oracle procedures.
\end{abstract}


\section{Introduction}
Suppose that for each of $n$ subjects we have measured a set of $p$ features which we want to use in a linear model to explain a  response, and  plan to use the Lasso to select relevant features. 
A desirable procedure controls the (expected) proportion of false
variables entering the explanatory model; that is, the ratio between the number of null coefficients estimated as nonzero by the Lasso procedure and the total number of coefficients estimated as nonzero. 
Conceptually, then, we imagine proceeding along the Lasso path by decreasing the penalty $\lambda$, and need to decide when to stop and reject the hypotheses corresponding to nonzero estimates, in such a way that the false discovery rate (FDR) does not exceed a preset level $q$. 

  The setup above differs from many classical
situations, where we know that the ``null statistics"---the statistics
for which the associated null hypothesis is true---are (approximately)
independent draws from a distribution that is (approximately) known;
in such situations, it is relatively easily to estimate the proportion
of null hypotheses that are rejected.  In stark contrast, in our
setting we do not know the sampling distribution of
$\widehat{\beta}_j(\lambda)$, and, to further complicate matters, this
distribution generally depends on the entire vector $\bbeta$ (this is
unlike the distribution of the least-squares estimator, where
$\widehat{\beta}^{\text{OLS}}_j$ depends only on $\beta_j$).

\citet{barber2015controlling} have proposed a method that utilizes artificial ``knockoff" features, and used it to circumvent the difficulties described above. 
The general idea behind the method is to introduce a set of ``fake" variables (i.e., variables that are known to not belong in the model), in such a way that there is some form of exchangeability under swapping a true null variable and a knockoff counterpart. 
This exchangeability essentially implies that a true null variable and a corresponding knockoff variable are indistinguishable (in the sense that the sufficient statistic is insensitive to swapping the two), hence given the event that one of the two was selected, each is equally probable to be the selected. 
We can then use the number of knockoff selections to estimate the
number of true nulls being selected. 

  In the original paper,
\citet{barber2015controlling} considered a Gaussian linear regression
model with a {\it fixed} number $p$ of fixed (nonrandom) predictors,
and $p\leq n$.  Since then, extensions and variations on knockoffs
have been proposed, which made the method applicable to a wide range
of problems.  \citet{barber2016knockoff} considered the high
dimensional ($n<p$) setting, where control over type III error (sign)
is also studied.  \citet{candes2016panning} constructed
``probabilistic'' model-X knockoffs for a nonparametric setting with
random covariates.  A group knockoff filter has been proposed by
\citet{dai2016knockoff} for detecting relevant groups of variables in
a linear model.  The focus in these works is on constructing knockoff
variables with the necessary properties to enable control over the
false discovery rate.  On the other hand, the Type II error rate has
been less carefully examined.

Power calculations are in general hard, especially when sophisticated test statistics, such as the Lasso estimates, are used. 
There is, however, a specific setting that enables us to exactly characterize (in some particular sense) the asymptotic distribution of the Lasso estimator relying on the theory of Approximate Message-Passing (AMP). 
More specifically, we work in the setup of
\citet{bayati2012lasso} which entails an $n\times p$ Gaussian design
matrix with $n,p\to \infty$ and $n/p\to \delta>0$; and the regression
coefficients are i.i.d.~copies from a mixture distribution
$\Pi = (1-\epsilon)\delta_{0} + \epsilon \Pi^*$, where $\Pi^*$ is the
distribution of the nonzero components\footnote{With some abuse of notation, in
  this paper we use $\Pi$ and $\Pi^*$ to refer to either a
  distribution or a {random variable} that has the corresponding
  distribution.} (i.e., has no point mass at zero).  We consider a
procedure that enters variables into the model according
to their order of (first) appearance on the Lasso path.
Appealing to the results in \citet{bayati2012lasso}, we identify an oracle procedure that, with the knowledge of the distribution $\Pi$ (and the noise level $\sigma^2$), is able to asymptotically attain maximum true positive rate subject to controlling the FDR. 
We show that a procedure which uses knockoffs for calibration, and hence a different solution path, is able to strictly control the FDR and at the same time asymptotically achieve power very close to that of the oracle \emph{adaptively in the unknown $\Pi$}.

Figure \ref{fig:intro} shows power of (essentially) the level-$q$ knockoff procedure proposed in Section \ref{sec:knockoffs} versus that of the level-$q$ oracle procedure, as $q$ varies and for two different distributions $\Pi^*$. 
The top two plots are theoretical asymptotic predictions, and the bottom plots show simulation counterparts. 
As implied above, the oracle gets to use the knowledge of $\Pi^*$ and $\epsilon$ in choosing the value of the Lasso penalty parameter $\lambda$, so that the asymptotic false discovery proportion is $q$; by contrast, the knockoff procedure is applied in exactly the same way in both situations, that is, it does not utilize any knowledge of $\Pi^*$ or $\epsilon$. 
Yet the top two plots in Figure \ref{fig:intro} show that the knockoff procedure adapts well to the unknown true distribution of the coefficients, at least for FDR levels in the range that would typically be of interest, in the sense that it blindly achieves close to optimal asymptotic power. 
Figures \ref{fig:tradeoff-exp} and \ref{fig:tradeoff-distributions} of Section \ref{fig:tradeoff-exp} give another representation of these findings, which are the main results of the paper.  

\begin{figure}[]
\centering
 \includegraphics[width=0.7\textwidth]{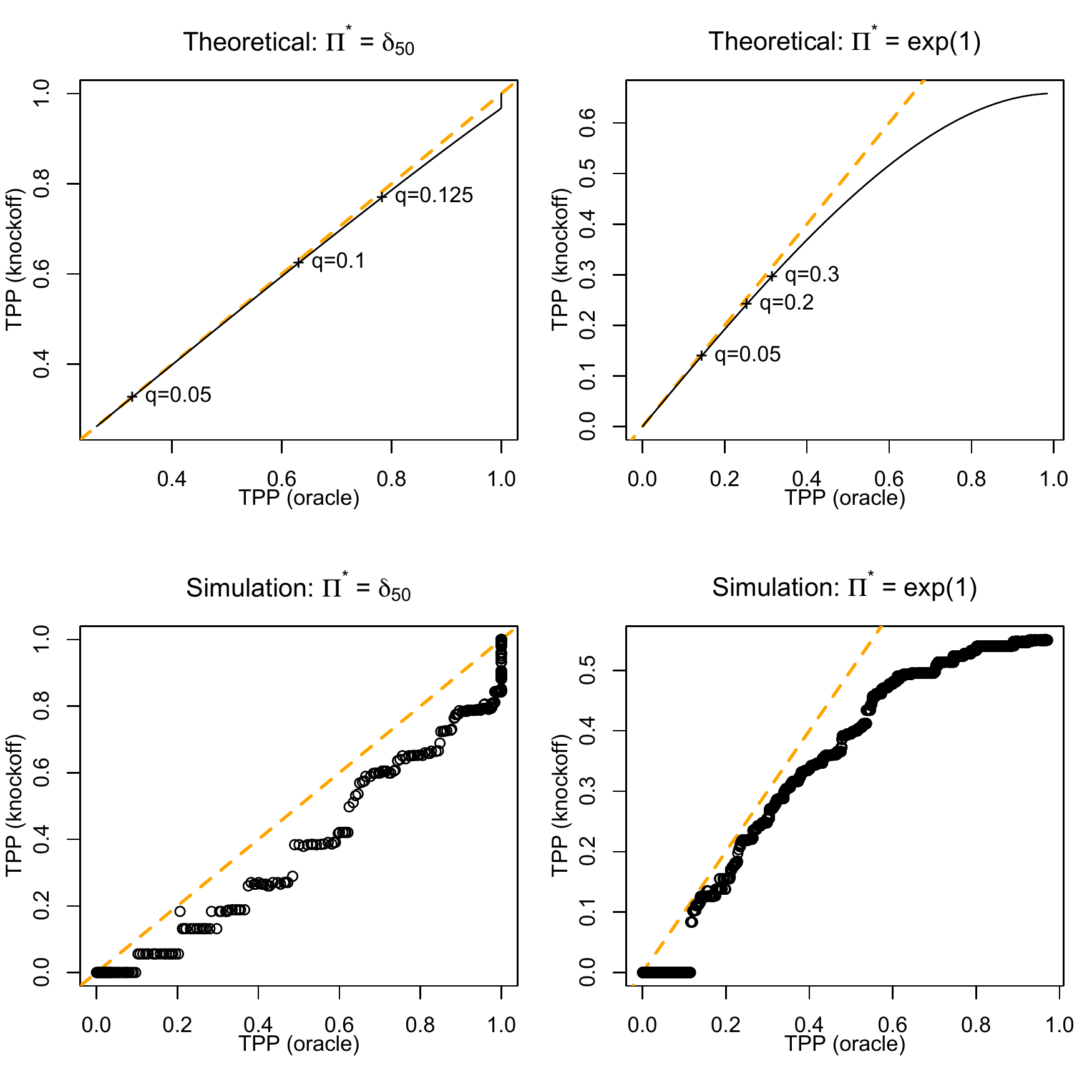}
 \caption{Asymptotic power of the knockoff procedure as compared to
   the oracle, for two different distributions $\Pi^*$: point mass
   (left) and exponential (right).
   $\epsilon = 0.2, \delta = 1, \sigma = 0.5$.  {\it Top}: Theoretical
   prediction.  Each of the points marked on the top graphs
   corresponds to a particular value of $q$, and was generated as
   follows.  (i) Fix $q$; (ii) find $\lambda$ and $t'$
   s.t. $\fdpinfty(t) = q$ and $\fdphatinftyprime(t') = q$, where
   $\fdpinfty$ and $\fdphatinftyprime$ are given in
   \eqref{eq:fdp-infty} and \eqref{eq:fdp-hat-infty-pi-naught-hat};
   (iii) plot the point
   $\left( \tppinfty(t),\tppinftyprime(t') \right)$, where $\tppinfty$
   and $\tppinftyprime$ are given in \eqref{eq:tpp-infty} and
   \eqref{eq:tpp-infty-prime}.  The broken line is the 45 degrees
   line.  The knockoff procedure adapts to an unknown $\Pi^*$: in both
   situations the asymptotic power (TPP) of the knockoffs procedure is
   very close to that of the oracle procedure.  {\it Bottom}:
   Simulation results, $\epsilon = 0.2$; $n=p=5000$; $\sigma = 0.5$.
   Each of the two graphs is meant to be an analogue of the
   asymptotic limits above.  }
\label{fig:intro}
\end{figure}

While the main focus is on the testing problem, we
address also the issue of model selection and prediction.  Motivated
by the work of \citet{abramovich2006special}, which formally connected
FDR control and estimation in the sequence model, our aim is to
investigate whether favorable properties associated with FDR control
carry over to the linear model.  As \citet{abramovich2006special}
remark at the outset, in the special case of the sequence model, a
full decision-theoretic analysis of the predictive risk is possible.
For the linear model case, \citet{bogdan2015slope} proposed a
penalized least squares estimator, SLOPE, {motivated} from the
perspective of multiple testing and FDR control, and
\citet{su2016slope} even proved that this estimator is asymptotically
minimax under the same design as that considered in the current paper,
although in a different sparsity regime, namely $\epsilon\to 0$.  In
truth, the linear model is substantially more difficult to analyze
than the sequence model.  Fortunately, however, the specific working
assumptions we adopt allow to analyze the (asymptotic) predictive
risk, again relying on the elegant results of \citet{bayati2012lasso} for
the Lasso estimates.  While our results lack the rigor of those in,
e.g., \citet{abramovich2006special} and \citet{su2016slope}, we think
that the findings reported in Section \ref{sec:prediction} are
valuable in investigating a problem for which conclusive answers are
currently not available.

The rest of the paper is organized as follows. 
Section \ref{sec:tradeoff} reviews some relevant results from \citet{su2017false} and presents the oracle false discovery rate-power tradeoff diagram for the Lasso. 
In Section \ref{sec:knockoffs} we present a simple testing procedure which uses knockoffs for FDR calibration, and prove some theoretical guarantees. 
A power analysis, which compares the knockoff procedure to the oracle procedure, is carried out in Section \ref{sec:power}. 
Section \ref{sec:prediction} is concerned with using knockoffs for prediction, and Section \ref{sec:discussion} concludes with a short discussion.

\section{An oracle tradeoff diagram}\label{sec:tradeoff}
Adopting the setup from \citet{su2017false}, we consider the linear
model 
\begin{equation}\label{eq:model}
\by = \bX\bbeta + \bz
\end{equation}
in which $\bX\in \RR^{n\times p}$ has i.i.d.~$\calN(0, 1/n)$ entries
and the errors $z_i$ are i.i.d.~$\calN(0,\sigma^2)$ with
$\sigma\geq 0$ fixed and arbitrary.  We assume that the regression
coefficients $\beta_j,\ j=1,...,p$ are i.i.d.~copies of a random
variable $\Pi$ with $\EE(\Pi^2)<\infty$ and
$\PP(\Pi\neq 0) = \epsilon\in (0,1)$ for a fixed constant $\epsilon$.
In this section we will be interested in the case where
$n,p\to \infty$ with $n/p \to \delta$ for a positive constant
$\delta$.  Note that $\Pi$ is assumed to not depend on $p$, and so the
expected number of nonzero elements $\beta_j$ is equal to
$\epsilon \cdot p$.

\noindent Let $\widehat{\bbeta}(\lambda)$ be the Lasso solution, 
\begin{equation}\label{eq:lasso}
\widehat{\bbeta}(\lambda) = \displaystyle \argmin_{\bb\in \RR^p} \frac{1}{2}\|\by - \bX\bb\|^2 + \lambda \|\bb\|_1,
\end{equation}
and denote $V(\lambda) = |\{j: \widehat{\beta}_j(\lambda) \neq 0, \beta_j = 0\}|$, $T(\lambda) = |\{j: \widehat{\beta}_j(\lambda) \neq 0, \beta_j \neq 0\}|$, $R(\lambda) = |\{j: \widehat{\beta}_j(\lambda)\neq 0\}|$ and $k = |\{j: \beta_j \neq 0\}|$. 
Hence $V(\lambda)$ is regarded as the number of false `discoveries' made by the Lasso; $T(\lambda)$ is the number of true discoveries; $R(\lambda)$ is the total number of discoveries, and $k$ is the number of true signals. 
We would like to remark here that $\beta_j = 0$ implies that $\bX_j$ (the $j$-th variable) is independent of $\by$ marginally {\it and} conditionally on any subset of $X_{-j}:=\{X_1,...,X_{j-1},X_{j+1},...,X_p\}$; hence the interpretation of rejecting the hypothesis $\beta_j = 0$ as a false discovery is clear and unambiguous. 
The false discovery proportion (FDP) is defined as usual as  
\begin{equation*}
\tFDP(\lambda) = \frac{V(\lambda)}{1\vee R(\lambda)}
\end{equation*}
and the true positive proportion (TPP) is defined as 
\begin{equation*}
\tTPP(\lambda) = \frac{T(\lambda)}{1\vee k}. 
\end{equation*}

\citet{su2017false} build on the results in \citet{bayati2012lasso} to devise a tradeoff diagram between TPP and FDP in the linear sparsity regime. 
Lemma A.1 in \citet{su2017false}, which is adopted from \citet{bogdan2013supplementary} and is a consequence of Theorem 1.5 in \citet{bayati2012lasso}, predicts the limits of FDP and TPP at a fixed value of $\lambda$. 
Throughout, let $\eta_t(\cdot)$ be the soft-thresholding operator, given by $\eta_{t}(x) = \mathrm{sgn}(x)(|x|-t)_+$. 
Also, let $\Pi^*$ denote a random variable distributed according to the conditional distribution of $\Pi$ given $\Pi \neq 0$; that is, 
\begin{equation}
\Pi = 
\begin{cases}
\Pi^*, &\text{w.p.} \ \epsilon, \\
0, &\text{w.p.} \ 1-\epsilon. 
\end{cases}
\end{equation}
Finally, denote by $\alpha_0$ the unique root of the equation
$(1+t^2)\Phi(-t) - t\phi(t) = \delta/2$.

\begin{lemma}[Lemma A.1 in \citealp{su2017false}; Theorem 1 in \citealp{bogdan2013supplementary}; see also Theorem 1.5 in \citealp{bayati2012lasso}]\label{lem:infty}
The Lasso solution with a fixed $\lambda>0$ obeys
\begin{equation*}
\begin{aligned}
\frac{V(\lambda)}{p}&\stackrel{\PP}{\to}2(1-\epsilon)\Phi(-\alpha),\\
\frac{T(\lambda)}{p}&\stackrel{\PP}{\to}\PP(|\Pi + \tau W|>\alpha\tau, \Pi \neq 0) = \epsilon \PP(|\Pi^* + \tau W|>\alpha\tau), 
\end{aligned}
\end{equation*}
where $W\sim \calN(0,1)$ independently of $\Pi$, and $\tau>0,\ \alpha >\max\{\alpha_0, 0\}$ is the unique solution to 
\begin{equation}\label{eq:system}
\begin{aligned}
\tau^2 &= \sigma^2 + \frac{1}{\delta}\EE(\eta_{\alpha\tau}(\Pi + \tau W) - \Pi)^2 \\
\lambda &= \left( 1 - \frac{1}{\delta}\PP(|\Pi + \tau W|>\alpha\tau) \right) \alpha\tau.
\end{aligned}
\end{equation}
\end{lemma}
As explained in \citealp{su2017false}, Lemma \ref{lem:infty} is a consequence of the fact that, under the working assumptions, $(\bbeta, \widehat{\bbeta}(\lambda))$ is in some (limited) sense asymptotically distributed as $(\bbeta, \eta_{\alpha\tau}(\bbeta + \tau\mathbf{W}))$, where $\mathbf{W} \sim\calN_p(\mathbf{0}, \mathbf{I})$ independently of $\bbeta$; here, the soft-thresholding operation acts on each component of a vector.  

It follows immediately from Lemma \ref{lem:infty} that for a fixed $\lambda > 0$, the limits of FDP and TPP are
\begin{equation}\label{eq:fdp-infty}
\tFDP(\lambda) \stackrel{\PP}{\longrightarrow} \frac{2(1-\epsilon)\Phi(-\alpha)}{2(1-\epsilon)\Phi(-\alpha) + \epsilon \PP(|\Pi^* + \tau W|>\alpha\tau)} \equiv \fdpinfty(\lambda)
\end{equation}
and
\begin{equation}\label{eq:tpp-infty}
\tTPP(\lambda) \stackrel{\PP}{\longrightarrow} \PP(|\Pi^* + \tau W|>\alpha\tau) \equiv \tppinfty(\lambda).
\end{equation}

\medskip 

The parametric curve $(\fdpinfty(\lambda), \tppinfty(\lambda))$ implicity defines the FDP level attainable at a specific TPP level for the Lasso solution with a fixed $\lambda$ as
\begin{equation}\label{eq:tradeoff-oracle}
q^{\Pi^*}(\tppinfty(\lambda); \epsilon, \delta, \sigma) = \fdpinfty(\lambda). 
\end{equation}
Interpreted in the opposite direction, for known $\Pi^*$ and
$\epsilon$, we have an {\it analytical expression} for the asymptotic
power of a procedure that, for a given $q$, chooses $\lambda$ in
\eqref{eq:lasso} in such a way that $\fdpinfty(\lambda) = q$. 
We describe how to compute the tradeoff curve $q^{\Pi^*}(\cdot ; \epsilon, \delta, \sigma)$ for a given $\Pi^*$ and $\epsilon$ in the Appendix. 

Because in practice $\Pi^*$ and $\epsilon$ would usually be unknown,
we refer to this procedure as the {\it oracle} procedure.  The oracle
procedure is not a ``legal'' multiple testing procedure, as it
requires knowledge of $\epsilon, \Pi^*$ and $\sigma^2$, which is of course realistically unavailable.  Still, it serves as a reference for evaluating
the performance of valid testing procedures.  Competing with this
oracle is the subject of the current paper, and the particular working
assumptions above were chosen simply because it is possible under
these assumptions to obtain a tractable expression for the empirical
distribution of the Lasso solution.

Before proceeding, we would like to mention another oracle procedure. 
Hence, consider the procedure which for $\lambda$ uses the value 
\begin{equation}
\lambda^* = \inf\{\lambda: \tFDP(\lambda)\leq q\}. 
\end{equation}
This can be thought of as an oracle that is allowed to see the actual
$\bbeta$ but is still committed to the Lasso path, and stops the last
``time" the FDP is smaller than $q$.  From now on, we think of
$\lambda$ as decreasing in time, and regard a variable $j$ as entering
{\it before} variable $j'$ if
$\sup \{\lambda: \widehat{\beta}_{j'}(\lambda)\neq 0\} \leq \sup
\{\lambda: \widehat{\beta}_{j}(\lambda)\neq 0\}$.
Although this may appear as a stronger oracle in some sense, it is
asymptotically equivalent to the ``fixed-$\lambda$" oracle discussed
earlier: by Theorem 3 in \citet{su2017false}, the event
\begin{equation}
\bigcap_{\lambda>0.01} \left\{ \tFDP(\lambda) \geq q^{\Pi^*}(\tTPP(\lambda) - 0.001; \epsilon, \delta, \sigma) \right\}
\end{equation}
has probability tending to one, hence the lower bound on FDP provided by $q^{\Pi^*}(\cdot; \epsilon, \delta, \sigma)$ holds uniformly in $\lambda$. 


\section{Calibration using knockoffs}\label{sec:knockoffs}
If we limit ourselves to procedures that use the Lasso with a fixed---or even data-dependent---value of $\lambda$, from the previous section we see that the achievable TPP can never asymptotically exceed that of the oracle procedure. 
Hence the asymptotic power of the oracle procedure serves as an asymptotic benchmark. 
In this section we propose a competitor to the oracle, which utilizes knockoffs for FDR calibration. 
While the procedure presented below is sequential and formally does not belong to the class of Lasso estimates with a data-dependent choice of $\lambda$, it will be almost equivalent to the corresponding procedure belonging to that class. 
We discuss this further later on. 

\begin{subsection}{A knockoff filter for the i.i.d.~design}\label{subsec:iid}
Let the model be given by \eqref{eq:model} as before. 
In this section, unless otherwise specified, it suffices to assume that the entries of $\bX$ are i.i.d.~from an arbitrary known continuous distribution $G$ (not necessarily normal). 
Furthermore, unless otherwise indicated, in this section we treat $n,p$ as fixed and condition throughout on $\bbeta$. 
{\it The definitions in the current section override the previous definitions, if there is any conflict. }

Let $\Xtil \in \RR^{n \times r}$ be a matrix with i.i.d.~entries drawn from $G$ independently of $\bX$ and of all other variables, where $r$ is a fixed positive integer. 
Next, let 
\begin{equation*}
\mathbb{X} := [\bX \ \Xtil] \in \RR^{n\times (p+r)}
\end{equation*}
denote the augmented design matrix, and consider the Lasso solution for the augmented design,
\begin{equation}\label{eq:lasso-til}
\widehat{\bbeta}(\lambda) = \displaystyle \argmin_{\bb\in \RR^{p+r}} \frac{1}{2}\|\by - \mathbb{X} \bb\|^2 + \lambda \|\bb\|_1.
\end{equation}
For simplicity we keep the notation $\widehat{\bbeta}(\lambda)$ for the solution corresponding to the augmented design, although it is of course different from \eqref{eq:lasso} (for one thing, \eqref{eq:lasso-til} has $p+r$ components versus $p$ for \eqref{eq:lasso}). 
Let
\begin{equation*}
\calH = \{1,...,p\},\ \ \ \calH_0 = \{j\in \calH: \beta_j = 0\},\ \ \ \calK_0 = \{p+1,...,p+r\}
\end{equation*}
be the sets of indices corresponding to the original variables, the null variables and the knockoff variables, respectively. 
Note that, because we are conditioning on $\bbeta$, the set $\calH_0$ is nonrandom. 
Now define the statistics
\begin{equation}\label{eq:stat}
T_j = \sup\{\lambda: \widehat{\beta}_j(\lambda) \neq 0\},\ \ j=1,...,p+r. 
\end{equation}
Informally, $T_j$ measures how early the $j$th variable enters the Lasso path (the larger, the earlier). 

Let 
\begin{equation*}
V_0(\lambda) = |\{j\in \calH_0: T_j \geq \lambda\}|,\ \ 
V_1(\lambda) = |\{j\in \calK_0: T_j \geq \lambda\}|
\end{equation*}
be the number of true null and fake null variables, respectively, which enter the Lasso path ``before" time $\lambda$. 
Also, let
\begin{equation*}
R(\lambda) = |\{j\in \calH: T_j \geq \lambda\}|
\end{equation*}
be the total number of original variables entering the Lasso path before ``time" $\lambda$. 
A visual representation of the quantities defined above is given in Figure \ref{fig:illustration}.

\begin{figure}[H]
\centering
 \includegraphics[width=.75\textwidth]{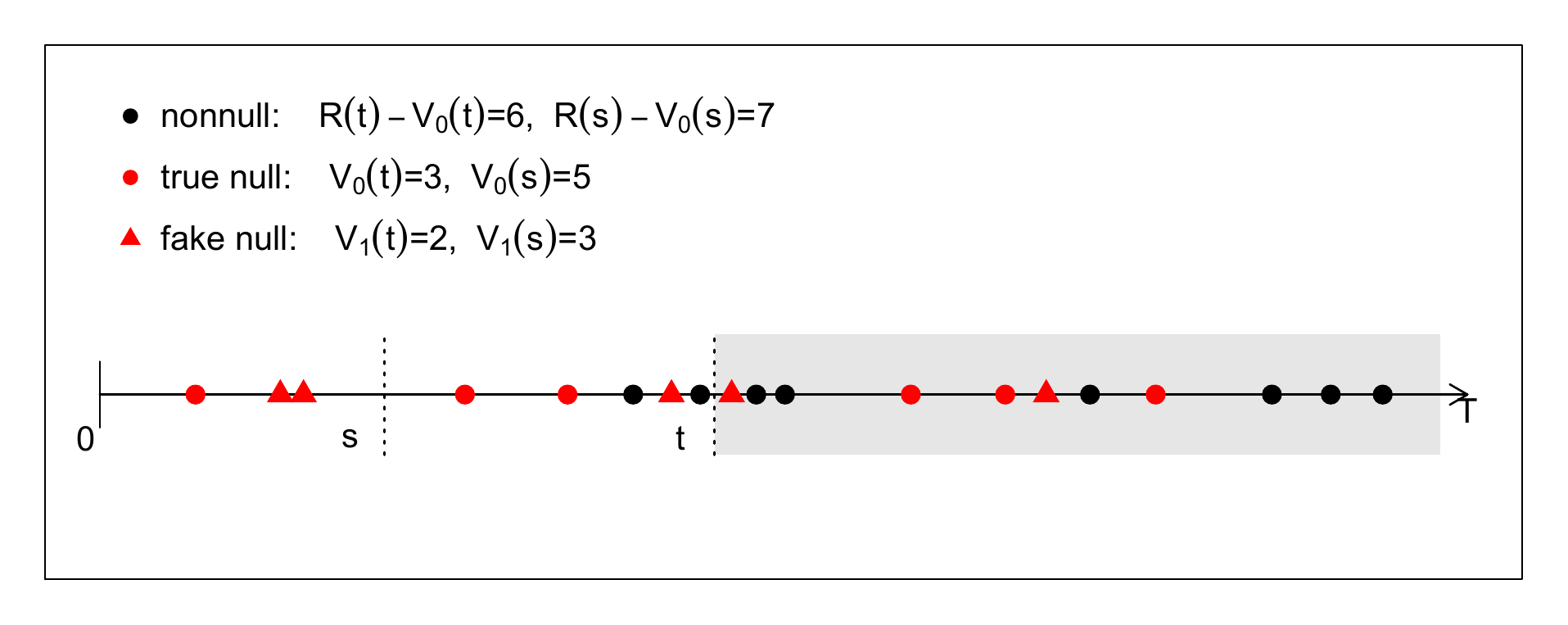}
\caption{
Representation of test statistics. 
Each $T_j$ is represented by a marker: round markers correspond to original variables ($\calH$), triangles represent fakes (knockoffs, $\calK_0$); black corresponds to nonnull ($\calH\setminus\calH_0$), red represents null (true or fake, $\calH_0\cup \calK_0$).  
}
\label{fig:illustration}
\end{figure}

The class of procedures we propose (henceforth, knockoff procedures) reject the null hypotheses corresponding to
\begin{equation*}
\{j \in \calH: T_j\geq \widehat{\lambda}\},
\end{equation*}
where 
\begin{equation}\label{eq:lambda-hat}
\widehat{\lambda} = \inf\left\{ \lambda \in \Lambda: \frac{ (1 + V_1(\lambda)) \cdot \frac{|\calH| \widehat{\pi}_0}{1+|\calK_0|} }{R(\lambda)} \leq q\right\}
\end{equation}
for an estimate $\widehat{\pi}_0$ of $\pi_0 := |\calH_0|/|\calH|$ and for a set $\Lambda$ associated with $\widehat{\pi}_0$. 
Note that conditionally on $\bbeta$, $\pi_0$ is a (nonrandom) parameter, and in this section we speak of estimating $\pi_0$ rather than $1-\epsilon$. 
Hence the tests in this class are indexed by the estimate of $\pi_0$ (and by $\Lambda$). 
Denote by $\mathbb{X}_j, \ j=1,...,p+r$, the columns of the augmented
matrix $\mathbb{X}$.  To see why the procedure makes sense, observe
that, conditionally on $\bbeta$, the distribution of
$(\mathbb{X}, \by)$ is invariant under permutations of the variables
in the set $\{\calH_0\cup \calK_0\}$ (just by symmetry).  Therefore,
conditionally on
$|\{j\in \calH_0\cup \calK_0: T_j \geq \lambda\}| = k$, any subset of
size $k$ of $\{\calH_0\cup \calK_0\}$ is equally probable to be the
selected set $\{j\in \calH_0\cup \calK_0: T_j \geq \lambda\}$.  It
follows that
\begin{equation*}
V_0(\lambda) / |\calH_0| \approx V_1(\lambda) / |\calK_0|. 
\end{equation*}
Therefore, 
\begin{equation*}
V_0(\lambda) \approx V_1(\lambda) \cdot \frac{|\calH| \pi_0}{|\calK_0|} \approx V_1(\lambda) \cdot \frac{|\calH| \widehat{\pi}_0}{|\calK_0|}
\end{equation*}
and hence \eqref{eq:lambda-hat} can be interpreted as the smallest $\lambda$ such that an estimate of 
\begin{equation*}
\tFDP(\lambda) \equiv \frac{ V_0(\lambda) }{ 1\vee R(\lambda)}
\end{equation*}
is below $q$. 
The fact that the procedure defined above, with $\widehat{\pi}_0 = 1$, controls the FDR, is a consequence of the more general result below. 

\begin{theorem}[FDR control for knockoff procedure with $\widehat{\pi}_0 = 1$]\label{thm:fdr-control}
  Let $T_j,\ j=1,...,p+r$ be test statistics with a joint distribution
  that is invariant to permutations in the set $\calH_0\cup \calK_0$
  (that is, reordering the statistics with indices in the extended
  null set, leaves the joint distribution unchanged).  Set
  $V_0(t) = |\{j\in \calH_0: T_j \geq t\}|$, 
  $V_1(t) = |\{j\in \calK_0: T_j \geq t\}|$, and
  $R(t) = |\{j\in \calH: T_j \geq t\}|$.  Then the procedure that
  rejects $\calH_0^j$ if $T_j\geq \tau$, where
\begin{equation}\label{eq:tau}
\tau = \inf\left\{ t \in \RR: \frac{ (1 + V_1(t)) \cdot \frac{|\calH| }{1+|\calK_0|} }{1\vee R(t)} \leq q\right\}
\end{equation}
controls the $\tFDR$ at level $q$. 
In fact,
$$
\EE \left[ \frac{ V_0(\tau) }{ 1\vee R(\tau)} \right] \leq q\frac{|\calH_0|}{|\calH|}. 
$$
\end{theorem}

\begin{proof}
We will use the following lemma, which is proved in the Appendix. 

\begin{lemma}\label{lem:hyper}
Let $X\sim \text{Hyper}(n_0, n_1; m)$ be a hypergeometric random variable, i.e., 
$$
\PP(X = k) = \frac{{n_0 \choose k} {n_1\choose m-k}}{{n_0+n_1 \choose m}}
$$
with $m>0$ and set $Y = X/(1+m-X)$. 
Then\footnote{Here and below we adopt the convention that ${0\choose 0} = 1$, and ${n\choose k} = 0$ if $n<0$ or if either $k<0$ or $k>n$.}
$$
\EE Y = \frac{n_0}{1+n_1} \left( 1 - \frac{{n_0-1 \choose m}}{{n_0+n_1 \choose m}} \right). 
$$
In particular,
$$
\EE Y \leq \frac{n_0}{1+n_1}
$$
for any value of $m$. 
\end{lemma}

We proceed to prove the theorem. 
Denote $m_0 = |\calH_0|$ and recall that $|\calH| = p$ and $|\calK_0| = r$. 
Recall also that we are conditioning throughout on $\bbeta$ (so that, in particular, $\calH_0$ is nonrandom). 
We will show that FDR control holds conditionally on $\bbeta$ (it therefore also holds unconditionally). 

With the notation above, 
$$
\tau = \inf\left\{ t: \frac{(1+V_1(t))\cdot \frac{p}{1+r}}{1\vee R(t)} \leq q \right\}.
$$
We have 
\begin{align*}
\tFDR = \EE \left[ \frac{V_0(\tau)}{1\vee R(\tau)} \right] 
= \EE \left[ \frac{(1+V_1(\tau))\cdot \frac{p}{1+r}}{1\vee R(\tau)} \cdot \frac{V_0(\tau)}{(1+V_1(\tau))\cdot \frac{p}{1+r}} \right]
&\leq q\EE\left[ \frac{V_0(\tau)}{(1+V_1(\tau)) \cdot \frac{p}{1+r}} \right]\\
&= q \frac{1+r}{p}\EE\left[ \frac{V_0(\tau)}{1+V_1(\tau)} \right].
\end{align*}
We will show that
\begin{equation}
\EE\left[ \frac{V_0(\tau)}{1+V_1(\tau)} \right] \leq \frac{m_0}{1+r}. 
\end{equation}
Consider the filtration 
$$
\calF_t = \sigma\left( \{V_0(s)\}_{0\leq s\leq t}, \{V_1(s)\}_{0\leq s\leq t}, \{R(s)\}_{0\leq s\leq t} \right). 
$$
That is, at all times $s\leq t$, we have knowledge about the number of true null, fake null and nonnull variables with statistics larger than $s$. 
We claim that 
$$
\frac{V_0(t)}{1+V_1(t)}
$$
is a supermartingale w.r.t. $\{\calF_t\}$, and $\tau$ is a stopping time. 
Indeed, let $t>s$. 
We need to show that 
\begin{equation}\label{eq:supermart}
\EE \left[\frac{V_0(t)}{1+V_1(t)} | \calF_s\right] \leq \frac{V_0(s)}{1+V_1(s)}. 
\end{equation}
The crucial observation is that, conditional on $V_0(s)$, $V_1(s)$ {\it and} $V_0(t) + V_1(t) = m$, 
$$
V_0(t )\sim \text{Hyper}(V_0(s), V_1(s); m). 
$$
This follows from the exchangeability of the test statistics for the true and fake nulls (in other words, conditional on everything else, the order of the red markers in Figure \ref{fig:illustration} is random). 

Now, Lemma \ref{lem:hyper} assures us that 
\begin{equation}
\EE \left[\frac{V_0(t)}{1+V_1(t)} | V_0(s), V_1(s), V_0(t) + V_1(t) = m\right] \leq \frac{V_0(s)}{1+V_1(s)}
\end{equation}
regardless of the value of $m$. 
Hence, the supermartingale property \eqref{eq:supermart} holds (and it is also clear that $\tau$ is a stopping time). 

By the optional stopping theorem for supermartingales, 
\begin{equation}\label{eq:optional}
\EE \left[\frac{V_0(t)}{1+V_1(t)} \right] \leq \EE \left[\frac{V_0(0)}{1+V_1(0)} \right] = \EE \left( \EE \left[\frac{V_0(0)}{1+V_1(0)} | V_0(0)+V_1(0) \right] \right). 
\end{equation}
But conditional on $V_0(0)+V_1(0)=m$ we have that $V_0(0 )\sim \text{Hyper}(m_0, r; m)$, 
and so, invoking Lemma \ref{lem:hyper} once again, we have
$$
\EE \left[\frac{V_0(0)}{1+V_1(0)} | V_0(0)+V_1(0)=m \right] \leq \frac{m_0}{1+r}
$$
no matter the value of $m$. 
\end{proof}

\end{subsection}

\begin{subsection}{Estimating the proportion of nulls}\label{subsec:pi-naught}

From our analysis, we see that FDR control would continue to hold if we replaced the number of hypotheses $p=|\calH|$ with the number of nulls $m_0=|\calH_0|$ in the definition of $\tau$. 
Of course, this quantity is unobservable, or, equivalently, $\pi_0 = |\calH_0|/|\calH|$ is unobservable. 
We can, however, estimate $\pi_0$ using a similar approach to that of
\citet{storey2002direct}, except that we replace calculations under
the theoretical null with the counting of knockoff variables at the ``bottom" of the Lasso path: specifically, for any $\lambda_0$ we have
$$
\frac{ |\{j\in \calH_0: T_j\leq \lambda_0\}| }{|\calH_0|} \approx \frac{ |\{j\in \calK_0: T_j\leq \lambda_0\}| }{|\calK_0|},
$$
again due to the exchangeability of null features. 
It follows that
$$
\begin{aligned}
|\calH_0| \approx |\calK_0| \cdot \frac{ |\{ j\in \calH_0:T_j\leq \lambda_0 \}| }{ |\{ j\in \calK_0:T_j\leq \lambda_0 \}| } &\leq (|\calK_0| + 1) \cdot \frac{ 1+|\{ j\in \calH_0:T_j\leq \lambda_0 \}| }{ |\{ j\in \calK_0:T_j\leq \lambda_0 \}| } \\
&\leq (|\calK_0| + 1) \cdot \frac{ 1+|\{ j\in \calH:T_j\leq \lambda_0 \}| }{ |\{ j\in \calK_0:T_j\leq \lambda_0 \}| }. 
\end{aligned}
$$
For small $\lambda_0$, we expect only few non-nulls among $\{j\in \calH: T_j\leq \lambda_0\}$ so that $|\{j\in \calH_0:T_j\leq \lambda_0\}| \approx |\{j\in \calH:T_j\leq \lambda_0\}|$. 
Hence, we propose to estimate $\pi_0$ by
\begin{equation}\label{eq:pi-naught-hat-lambda}
\widehat{\pi}_0 = \frac{(|\calK_0| + 1)}{|\calH|} \cdot \frac{ 1+|\{ j\in \calH:T_j\leq \lambda_0 \}| }{ |\{ j\in \calK_0:T_j\leq \lambda_0 \}| }. 
\end{equation}
We again state our result more generally.

\begin{theorem}[FDR control for knockoff procedure with an estimate of $\widehat{\pi}_0$]\label{thm:fdr-control-est}
  Set \begin{equation}\label{eq:pi-naught-hat}
\widehat{\pi}_0 = \frac{(|\calK_0| + 1)}{|\calH|} \cdot \frac{ 1+|\{ j\in \calH:T_j\leq t_0 \}| }{ |\{ j\in \calK_0:T_j\leq t_0 \}| } 
\end{equation}
($\widehat{\pi}_0$ is set to $\infty$ if the denominator vanishes). 
  In the setting of Theorem \ref{thm:fdr-control}, the procedure that
  rejects $\calH_0^j$ if $T_j\geq \tau$, where
\begin{equation}\label{eq:tau-est}
\tau = \inf\left\{ t \geq t_0: \frac{ (1 + V_1(t)) \cdot \frac{|\calH| \widehat{\pi}_0}{1+|\calK_0|} }{1\vee R(t)} \leq q \right\}
\end{equation}
obeys
$$
\tFDR \leq q.
$$

\end{theorem}

\begin{proof}
We use the same notation as in the proof of Theorem \ref{thm:fdr-control}. 
As before, we have that 
$$
\tFDR \leq q\EE\left[ \frac{V_0(\tau)}{(1+V_1(\tau)) \cdot \frac{p \widehat{\pi}_0}{1+r}} \right].
$$
Again, the process $\{V_0(t)/(1+V_1(t))\}_{t\geq t_0}$
starting at $t_0$, is a supermartingale w.r.t.
$\{\calF_t\}_{t\geq t_0}$ and $\tau$ is a stopping time.  Then the
optional stopping theorem gives
\begin{equation}
\begin{aligned}
\EE\left[ \frac{V_0(\tau)}{(1+V_1(\tau)) \cdot \frac{p \widehat{\pi}_0}{1+r}} \right] &\leq \EE\left[ \frac{V_0(t_0)}{(1+V_1(t_0)) \cdot \frac{p \widehat{\pi}_0}{1+r}} \right]\\
&=\EE\left[ \frac{V_0(t_0)}{1+V_1(t_0)} \cdot \frac{r-V_1(t_0)}{1+p-R(t_0)} \right]\\
&\leq \EE \left[ \frac{V_0(t_0)}{1+V_1(t_0)} \cdot
  \frac{r-V_1(t_0)}{1+m_0-V_0(t_0)} \right]. 
\end{aligned}
\end{equation}
As before, $V_0(t_0) \sim \text{Hyper}(m_0, r;m)$ conditional on $V_0(t_0) + V_1(t_0)=m$. 
In the Appendix we prove that if a random variable $X$ follows this distribution, then 
\begin{equation}\label{eq:hypergeometric-identity}
\EE \left[  \frac{X}{1+m-X} \cdot  \frac{r+X-m}{1+m_0-X} \right] = 1 - \frac{ {m_0\choose m_0\wedge m} {r\choose m-m_0\wedge m} }{ {m_0+r\choose m} } 
\end{equation}
In particular, for all $m$ the result is no more than $1$, thereby completing the proof of the theorem. 

\end{proof}

In practice, we of course suggest to use the minimum between one and
\eqref{eq:pi-naught-hat-lambda} as an estimator for $\pi_0$.  While we
cannot obtain a result similar to Theorem \ref{thm:fdr-control-est}
when the truncated estimate of $\pi_0$ is used instead of
\eqref{eq:pi-naught-hat-lambda}, we can appeal to results from
approximate message passing, in the spirit of the calculations from
Section \ref{sec:power}, to obtain an approximate asymptotic result.
Thus, let
\begin{equation*}
\widehat{\pi}_0 = 1 \ \wedge\ \left( \frac{(|\calK_0| + 1)}{|\calH|} \cdot \frac{ 1+|\{ j\in \calH:T_j\leq t_0 \}| }{ |\{ j\in \calK_0:T_j\leq t_0 \}| }\right).
\end{equation*}
and adopt the working assumptions stated at the beginning of Section \ref{sec:power}. 
As in the proof of the theorem above, 
$$
\tFDR \leq q\EE\left[ \frac{V_0(\tau)}{(1+V_1(\tau)) \cdot \frac{p \widehat{\pi}_0}{1+r}} \right],
$$
and by the same reasoning as above, 
\begin{equation}
\begin{aligned}
\EE\left[ \frac{V_0(\tau)}{(1+V_1(\tau)) \cdot \frac{p \widehat{\pi}_0}{1+r}} \right] &\leq \EE\left[ \frac{V_0(t_0)}{(1+V_1(t_0)) \cdot \frac{p \widehat{\pi}_0}{1+r}} \right]\\
&=\EE\left[ \frac{V_0(t_0)}{1+V_1(t_0)} \cdot \max \left( \frac{r-V_1(t_0)}{1+p-R(t_0)}, \frac{1+r}{p} \right) \right]\\
&\leq \EE \left[ \frac{V_0(t_0)}{1+V_1(t_0)} \cdot \max \left( \frac{r-V_1(t_0)}{1+m_0-V_0(t_0)}, \frac{1+r}{p} \right) \right]. 
\end{aligned}
\end{equation}
Now, as in equation \eqref{eq:fdp-prime} below, we have 
$$
\frac{V_0(t_0)}{p} = \frac{|\{j\in \calH_0: T_j \geq t_0\}|}{p} \approx \frac{|\{j\in \calH_0: \widehat{\beta}_j(t_0)\neq 0, \beta_j = 0\}|}{p} \to 2(1-\epsilon')\Phi(-\alpha_1'), 
$$
where $\epsilon' = \epsilon/(1 + \rho)$ and where $\alpha_1'$ is the same as in equation \eqref{eq:fdp-hat-infty-pi-naught-hat}, when replacing $\lambda_0$ with $t_0$. 
Also, using equation \eqref{eq:fdp-fake-prime}, we have
$$
\frac{V_1(t_0)}{r} = \frac{|\{j\in \calK_0: T_j \geq t_0\}|}{r} \approx \frac{|\{j\in \calK_0: \widehat{\beta}_j(t_0)\neq 0\}|}{r} \to 2\Phi(-\alpha_1'). 
$$
Treating the approximations above as equalities (in Section \ref{sec:power} we argue that these should be good approximations), we have
$$
\frac{V_0(t_0)}{1+V_1(t_0)}\bigg/\frac{p}{r} \to \frac{2(1-\epsilon')\Phi(-\alpha_1')}{2\Phi(-\alpha_1')} = 1-\epsilon'
$$
and
$$
\frac{r-V_1(t_0)}{1+m_0-V_0(t_0)}\bigg/\frac{r}{p} \to \frac{[1-2\Phi(-\alpha_1')]}{(1-\epsilon')[1-2\Phi(-\alpha_1')]} = \frac{1}{1-\epsilon'},\ \ \ \ \ \ 
\frac{1+r}{p}\bigg/\frac{r}{p} \to 1. 
$$
Hence,
$$
\frac{V_0(t_0)}{1+V_1(t_0)} \cdot \max \left( \frac{r-V_1(t_0)}{1+m_0-V_0(t_0)}, \frac{1+r}{p} \right) \to (1-\epsilon') \cdot \max(\frac{1}{1-\epsilon'}, 1) = 1. 
$$
We conclude that, up to the approximations above, 
$$
\limsup_{n,p\to \infty} \tFDR \leq q.
$$

\end{subsection}

\section{Power Analysis}\label{sec:power}
The main observation in this paper is that we can use the consequences of Theorem 1.5 in \citet{bayati2012lasso} to obtain an asymptotic TPP-FDP tradeoff diagram for the Lasso solution associated with the {\it knockoff} setup; that is, the Lasso solution corresponding to the {augmented} design $\mathbb{X}$. 
Indeed, the working assumptions required for applying Lemma \ref{sec:tradeoff} continue to hold, only with a slightly different set of parameters. 
Recall that $\mathbb{X}\in \RR^{n\times(p+r)}$ has i.i.d.~$\calN(0,1/n)$ entries. 
Taking $r = \ceil{\rho p}$ for a constant $\rho>0$, we have that $n/(p+r) \to \delta' := \delta/(1+\rho)$ as $n,p\to \infty$. 
Furthermore, because $\EE Y$ is related to $\mathbb{X}$ through 
$$
\EE Y = \mathbb{X}
\begin{bmatrix}
\bbeta \\
\mathbf{0}
\end{bmatrix},
$$
where $\mathbf{0} = (0,....,0)^T\in \RR^r$, we have that the {empirical distribution} of the coefficients corresponding to the {\it augmented} design converges exactly to
\begin{equation}\label{eq:pi-prime}
\Pi' = 
\begin{cases}
\Pi^*, &\text{w.p.} \ \epsilon',\\
0, &\text{w.p.} \ 1-\epsilon',
\end{cases}
\end{equation}
where $\epsilon' := \epsilon/(1 + \rho) = \lim \{p\epsilon/(p+\ceil{\rho p})\}$. 
Fortunately, it is only the {\it empirical} distribution of the coefficient vector that needs to have a limit of the form \eqref{eq:pi-prime}, for Lemma \ref{lem:infty} to follow from Theorem 1.5 of \citet{bayati2012lasso}; see also the remarks in Section 2 of \citet{su2017false}. 
The working hypothesis of Section \ref{sec:tradeoff} therefore holds with $\delta$ and $\epsilon$ replaced by $\delta'$ and $\epsilon'$. 
Recalling the definitions of $\calH, \calH_0$ and $\calK_0$ from Section \ref{sec:knockoffs}, a counterpart of Lemma \ref{lem:infty} asserts that
\begin{align}
\frac{|\{j\in \calH: \widehat{\beta}_j(\lambda)\neq 0, \beta_j = 0\}|}{p} &\longrightarrow 2(1-\epsilon)\Phi(-\alpha'), \\ \label{eq:fdp-prime}
\frac{|\{j\in \calH: \widehat{\beta}_j(\lambda)\neq 0, \beta_j \neq 0\}|}{p} &\longrightarrow \epsilon \PP (|\Pi^* + \tau'W|>\alpha'\tau')
\end{align}
and
\begin{equation}\label{eq:fdp-fake-prime}
\frac{|\{j\in \calK_0: \widehat{\beta}_j(\lambda)\neq 0\}|}{r} \longrightarrow 2\Phi(-\alpha'), 
\end{equation} 
where $\widehat{\beta}_j(\lambda)$ is the solution to \eqref{eq:lasso-til}; $W$ is a $\calN(0,1)$ variable independent of $\Pi$; and $\tau'>0,\ \alpha' >\max\{\alpha_0', 0\}$ is the unique solution to \eqref{eq:system} when $\delta$ is replaced by $\delta'$ and $\epsilon$ is replaced by $\epsilon'$. 
Hence, 
\begin{equation}\label{eq:fdp-infty-prime}
\begin{aligned}
\tFDP(\lambda) &= \frac{|\{j\in \calH_0: T_j \geq \lambda\}|}{1\vee |\{j\in \calH: T_j \geq \lambda\}|} \\
&= \frac{|\{j\in \calH_0: \widehat{\beta}_j(\lambda')\neq 0 \text{ for some } \lambda' \geq \lambda\}|}{1\vee |\{j\in \calH: \widehat{\beta}_j(\lambda')\neq 0 \text{ for some } \lambda' \geq \lambda\}|}\\
&\approx \frac{|\{j\in \calH_0: \widehat{\beta}_j(\lambda)\neq 0\}|}{1\vee |\{j\in \calH: \widehat{\beta}_j(\lambda)\neq 0\}|} \\
&\longrightarrow \frac{2(1-\epsilon)\Phi(-\alpha')}{2(1-\epsilon)\Phi(-\alpha') + \epsilon \PP (|\Pi^* + \tau'W|>\alpha'\tau')} 
\equiv \fdpinftyprime(t).
\end{aligned}
\end{equation}
Above, we allowed ourselves to approximate
$|\{j\in \calH: \widehat{\beta}_j(\lambda')\neq 0 \text{ for some }
\lambda' \geq \lambda\}|$ with
$|\{j\in \calH: \widehat{\beta}_j(\lambda)\neq 0\}|$, and likewise for
$\calH_0$.  If least-angle regression \citep{efron2004least} 
were
used instead of the Lasso, the two quantities would have been equal;
but because the Lasso solution path is considered here, we cannot
completely rule out the event in which a variable that entered the
path drops out, and the former quantity is in general larger.  Still,
we expect the difference to be very small, at least for large enough
values of $\lambda$ (corresponding to small FDP).  
  This is verified by the simulation of Figure
\ref{fig:simulations}: the green circles, computed using the
statistics $T_j$, are very close to the red circles, computed when
rejections correspond to variables with
$\widehat{\beta}_j(\lambda)\neq 0$.  All three simulation examples
were computed for a single realization on the original $n\times p$ design (not the augmented
design). Similar results were obtained when we repeated the experiment (i.e., for a different realization of the data). 

\begin{figure}[]
\centering
 \includegraphics[width=.8\textwidth]{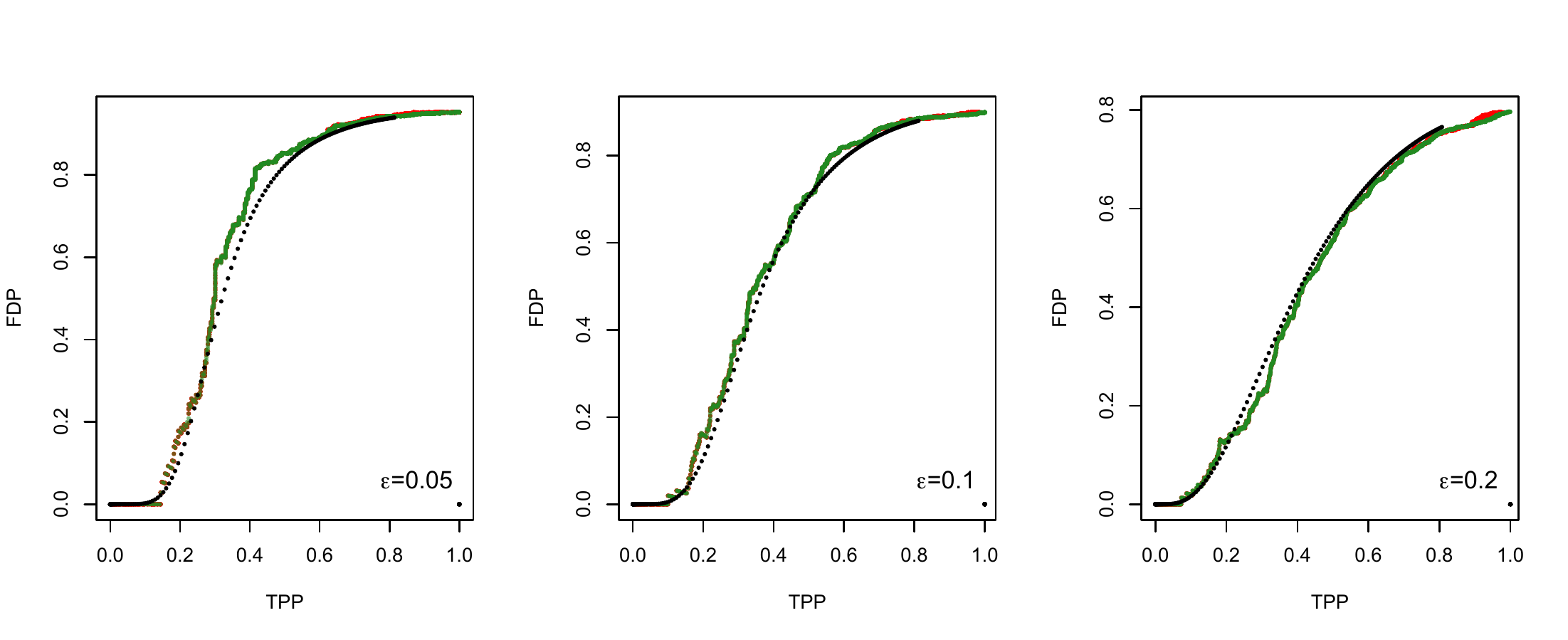}
 \caption{ Discrepancy between testing procedures for three different
   values of $\epsilon$.  In all three examples $n=p=5000$,
   $\Pi^*\sim \exp(1)$, and $\sigma = 0.5$; and TPP is plotted against
   FDP.  Green circles correspond to the procedure that uses the
   statistics $T_j$ in \eqref{eq:stat}.  Red circles correspond to the
   procedure that rejects when $\widehat{\beta}_j(\lambda)\neq 0$.
   Black circles are theoretical predictions from Section
   \ref{sec:power} with $\delta = 1$.  Each curve is is computed from
   a single realization of the data, and uses an $n\times p$ design
   matrix (i.e., not the augmented design). 
       }
\label{fig:simulations}
\end{figure}

\noindent Similarly, 
\begin{equation}\label{eq:tpp-infty-prime}
\begin{aligned}
\tTPP_{\text{aug}}(\lambda) &\equiv \frac{|\{j\in \calH\setminus \calH_0: T_j \geq \lambda\}|}{1\vee |\{j\in \calH: \beta_j \neq 0\}|} \\
&\approx \frac{|\{j\in \calH \setminus \calH_0: \widehat{\beta}_j(\lambda)\neq 0\}|}{1\vee |\{j\in \calH: \beta_j \neq 0\}|}\\
&\longrightarrow \PP (|\Pi^* + \tau'W|>\alpha'\tau') \equiv \tppinftyprime(\lambda)
\end{aligned}
\end{equation}
and
\begin{equation}\label{eq:fdp-hat-infty-gen}
\begin{aligned}
\widehat{\tFDP}_{\text{aug}}(\lambda) &\equiv \frac{ (1 + |\{j\in \calK_0: T_j \geq \lambda\}|) \cdot \frac{|\calH| \widehat{\pi}_0}{1+|\calK_0|} }{|\{j\in \calH: T_j \geq \lambda\}|} \approx 
\frac{ (1 + |\{j\in \calK_0: \widehat{\beta}_j(\lambda)\neq 0\}|) \cdot \frac{|\calH| \widehat{\pi}_0}{1+|\calK_0|} }{|\{j\in \calH: \widehat{\beta}_j(\lambda)\neq 0\}|}\\
&\longrightarrow \frac{2\Phi(-\alpha')}{2(1-\epsilon)\Phi(-\alpha') +
  \epsilon \PP (|\Pi^* + \tau'W|>\alpha'\tau')} \cdot \lim
\widehat{\pi}_0, 
\end{aligned}
\end{equation}
where $W$, $\tau'>0,\ \alpha' >\max\{\alpha_0', 0\}$ are as before, and where $\widehat{\pi}_0$ is any estimate of $\pi_0$ that has a limit. 
Taking $\widehat{\pi}_0$ to be the minimum between one and \eqref{eq:pi-naught-hat}, we have that the limit in \eqref{eq:fdp-hat-infty-gen} can be approximated by 
\begin{equation}\label{eq:fdp-hat-infty-pi-naught-hat}
\begin{aligned}
&\frac{2\Phi(-\alpha')}{2(1-\epsilon)\Phi(-\alpha') + \epsilon \PP (|\Pi^* + \tau'W|>\alpha'\tau')} \\
&\cdot \left(1 \wedge \left\{ 1-\epsilon + \frac{ \epsilon [\PP (|\Pi^* + \tau'_2 W|>\alpha'_2\tau'_2) - \PP (|\Pi^* + \tau'_1 W|>\alpha'_1\tau'_1)] }{2[\Phi(-\alpha'_2) - \Phi(-\alpha'_1)]} \right\}\right) \equiv \fdphatinftyprime(\lambda),
\end{aligned}
\end{equation}
where $\tau'_1>0,\ \alpha'_1 >\max\{\alpha_0', 0\}$ is the unique
solution to \eqref{eq:system} when $\lambda = \lambda_0$ and when
$\delta$ is replaced by $\delta'$ and $\epsilon$ is replaced by
$\epsilon'$; and where $\tau'_2>0,\ \alpha'_2 >\max\{\alpha_0', 0\}$
is the limit of the unique solution to \eqref{eq:system} when $\lambda
\to 0^+$ also with $\delta'$ (resp.~$\epsilon'$) in lieu of $\delta$ (resp.~$\epsilon$).
Indeed,  
\begin{equation}\label{eq:tau-naught-I}
\begin{aligned}
\frac{|\{j\in \calH: 0<T_j< \lambda_0\}|}{p} &= \frac{|\{j\in \calH: T_j>0\}| - |\{j\in \calH: T_j\geq \lambda_0\}|}{p}\\
&\approx \frac{|\{j\in \calH: \widehat{\beta}_j(0^+)\neq 0\}| - |\{j\in \calH: \widehat{\beta}_j(\lambda_0)\neq 0\}|}{p}\\
&\longrightarrow 2(1-\epsilon)\Phi(-\alpha'_2) + \epsilon \PP (|\Pi^* + \tau'_2 W|>\alpha'_2\tau'_2)\\
&- 2(1-\epsilon)\Phi(-\alpha'_1) + \epsilon \PP (|\Pi^* + \tau'_1 W|>\alpha'_1\tau'_1) 
\end{aligned}
\end{equation}
and, by a similar calculation, 
\begin{equation}\label{eq:tau-naught-II}
\begin{aligned}
\frac{|\{j\in \calK_0: 0<T_j\leq \lambda_0\}|}{|\calK_0|} \longrightarrow 2\Phi(-\alpha'_2) - 2\Phi(-\alpha'_1);
\end{aligned}
\end{equation}
using \eqref{eq:tau-naught-I} and \eqref{eq:tau-naught-II} to compute the limit of \eqref{eq:pi-naught-hat}, we obtain \eqref{eq:fdp-hat-infty-pi-naught-hat}.

\begin{figure}[]
\centering
 \includegraphics[width=.8\textwidth]{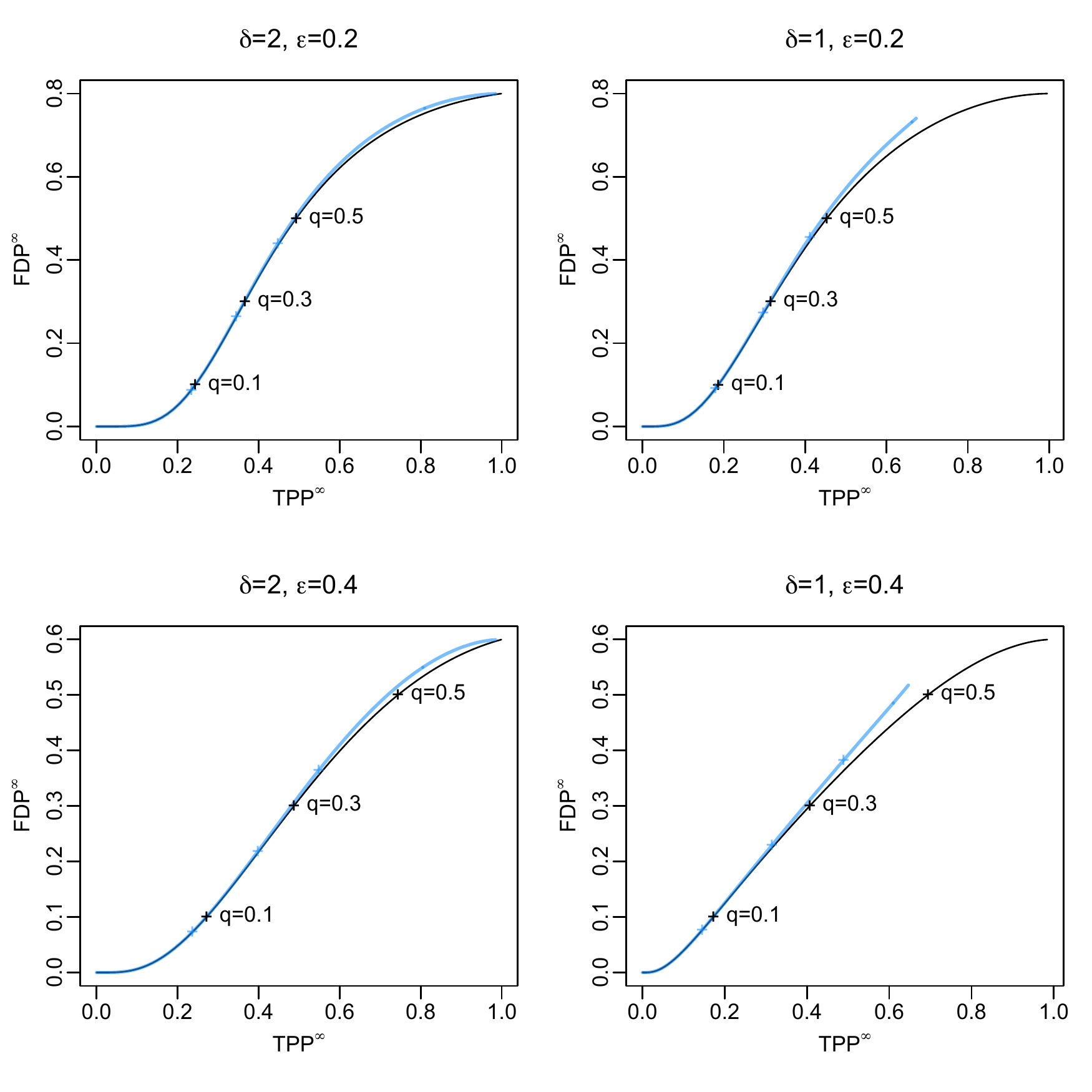}
 \caption{
     Tradeoff diagrams for
   $\Pi^*\sim \exp(1)$.  $\sigma = 0.5$; $\rho=1$.  Different panels
   correspond to different combinations of $\delta$ and $\epsilon$, as
   appears in the title of each of the plots.  The black (resp.~light
   blue) curve corresponds to the oracle (resp.~knockoffs).  Markers
   denote the pair $(\mathrm{tpp},\mathrm{fdp})$ of asymptotic power
   and type I error, attained for three example values of $q$: black
   for oracle, light blue for knockoff.  The knockoff procedure loses
   a little bit of power due to the estimate of $1-\epsilon$
   (proportion of true nulls).  }
\label{fig:tradeoff-exp}
\end{figure}

As $\lambda>0$ varies, a parametric curve $(\fdpinftyprime(\lambda), \tppinftyprime(\lambda))$ is traced which specifies the FDP level attainable at a specific TPP level for the Lasso solution corresponding to the {\it augmented} design, as
\begin{equation}
q_{\text{aug}}^{\Pi^*}(\tppinftyprime(\lambda); \epsilon, \delta, \sigma) = \fdpinftyprime(\lambda). 
\end{equation}

Figure \ref{fig:tradeoff-exp} shows this curve against the analogous curve $(\fdpinfty(\lambda), \tppinfty(\lambda))$, associated with the original design $\bX$ and the oracle procedure, when $\Pi^*$ is exponential with mean $1$. 
The different panels correspond to different combinations of $\delta$ and $\epsilon$. 
The proportion $\rho$ was taken to be $1$ in all scenarios, in other words, the number of knockoff variables $r$ was taken to be equal to the number of original variables $p$; we observed only very slight differences in the curves when $\rho$ was varied between $0.1$ and $1$. 
Figure \ref{fig:tradeoff-distributions} shows the two curves, $(\fdpinfty(\lambda), \tppinfty(\lambda))$ and $(\fdpinftyprime(\lambda), \tppinftyprime(\lambda))$, for different distributions $\Pi^*$.

\begin{figure}[]
\centering
 \includegraphics[width=.8\textwidth]{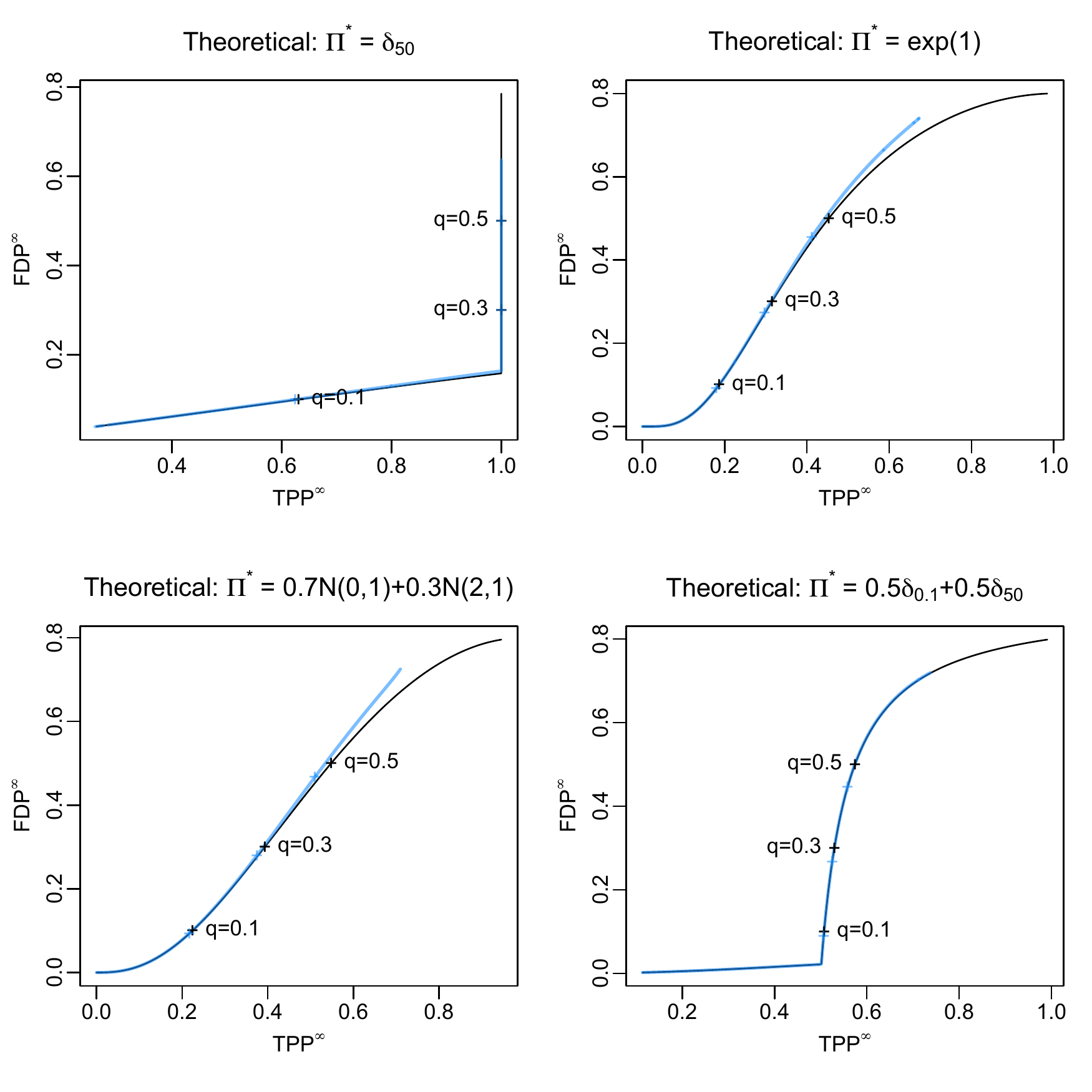}
\caption{
Tradeoff diagrams for different $\Pi^*$, with $\epsilon = 0.2$, $\delta = 1$ and $\sigma = 0.5$. 
}
\label{fig:tradeoff-distributions}
\end{figure}

In all situations depicted in Figures \ref{fig:tradeoff-exp} and \ref{fig:tradeoff-distributions}, the light blue curve, corresponding to the augmented setup, lies very close to the black curve, which corresponds to the original design. 
Hence, augmenting the design with knockoffs seems to have little effect on the tradeoff diagram, or, the achievable power for a given FDP level, at least for FDP levels small enough to be of interest. 
For example, if we set $q=0.05$, then there exist values $\lambda$ and $t'$ such that $\fdpinfty(\lambda) = \fdpinftyprime(\lambda') = 0.05$ and $\tppinfty(\lambda) \approx \tppinftyprime(\lambda')$. 
Of course, the knockoff procedure achieves this without knowledge of $\epsilon$ and $\Pi^*$, while the oracle relies on knowing these quantities in choosing $\lambda$. 
Note, still, that for a given input $q$, the knockoff procedure essentially (asymptotically) chooses $\lambda'$ such that $\fdphatinftyprime(\lambda')$, not $\fdpinftyprime(\lambda')$, is equal to $q$. 
Because the former uses a conservative estimate of $1-\epsilon$, $\fdphatinftyprime(\lambda')$ is usually larger than $\fdpinftyprime(\lambda')$ and, consequently, the power attained by the knockoff procedure is a little lower. 
This gap is depicted by the ``+" markers in the plots, each pair of blue and black markers corresponding to a particular value of $q$. 
For example, for $\delta = 1$ and $\epsilon = 0.2$, oracle power is $0.187$ and knockoffs power is $0.18$ at $q = 0.1$. 
The value of $t_0$ used is $0.1$. 
Figure \ref{fig:intro} gives an alternative representation (see caption). 

\section{Prediction}\label{sec:prediction}

Armed with a procedure that controls the FDR for the model
\eqref{eq:model}, we now explore utilizing knockoffs for estimation
purposes (the prediction error and the estimation error are equivalent
here because we have uncorrelated predictors).  Our pursuit is to a
large degree inspired by the work of \citet{abramovich2006special},
which for the first time rigorously linked FDR control and estimation
error.
At the same time, the working assumptions we adopt in the current paper allow us to take further advantage of the results of \citet{bayati2012lasso}, namely, to asymptotically analyze not only the power or the knockoff procedure but also the {\it estimation} error associated with it. 
As in the preceding analysis of testing errors, the (normalized) risk of the Lasso estimator at a fixed $\lambda$ tends in probability to a simple limit. 
Specifically, by Corollary 1.6 of \citet{bayati2012lasso}, we have that, for a given pair of $\epsilon$ and $\Pi^*$, and for fixed $\lambda$, 
\begin{equation}
\lim \frac{1}{p} \| \widehat{\bbeta}(\lambda) - \bbeta \|^2 \longrightarrow \EE\left[ \left(\eta_{\alpha\tau}(\Pi + \tau W) - \Pi \right)^2 \right] \equiv R^{\infty}(\lambda; \Pi^*, \epsilon),
\end{equation}
where $W\sim \calN(0,1)$ independently of $\Pi$, and $\tau>0,\ \alpha >\max\{\alpha_0, 0\}$ is the unique solution to \eqref{eq:system}. 
Because we have a tractable form for the limiting risk, we can minimize it over $\lambda$ to obtain an {\it oracle} choice for $\lambda$,
\begin{equation}\label{eq:oracle-lambda}
\lambda^{\infty}_{\text{OL}}(\Pi^*, \epsilon) = \argmin_{\lambda} R^{\infty}(\lambda; \Pi^*, \epsilon). 
\end{equation}
The oracle choice of $\lambda$ depends on $\Pi^*$ and on $\epsilon$ and therefore does not produce a legitimate estimator for $\bbeta$, but it sets a benchmark for evaluating the performance of a procedure that plugs in any (possibly data-dependent) value for $\lambda$. 
Consider now a fixed $q$ and the knockoff procedure from Section \ref{subsec:pi-naught}, which selects the $j$-th variable if $T_j\geq \tau$, where $\tau$ is given by \eqref{eq:lambda-hat} and $\widehat{\pi}_0$ by \eqref{eq:pi-naught-hat} and $\Lambda = [t_0,\infty)$. 
This is a legal variable selection procedure (it does not depend on $\Pi^*$ or on $\epsilon$) with the FDR control guarantees of Theorem \ref{thm:fdr-control-est}. 
In the limit as $n,p\to \infty$, it corresponds to a choice of $\lambda$ defined by 
\begin{equation*}
\fdphatinftyprime(\lambda) = q
\end{equation*}
where $\fdphatinftyprime(\lambda)$ is given in \eqref{eq:fdp-hat-infty-pi-naught-hat}. 
Call this the ``knockoff" choice for $\lambda$ and denote it by $\lambda^{\infty}_{\text{KO}}(q)$. 
Then the quantity
\begin{equation}\label{eq:inflation}
\nu(q) \equiv \sup_{\Pi^*} \left\{ \frac{R^{\infty}(\lambda^{\infty}_{\text{KO}}(q); \Pi^*, \epsilon)}{R^{\infty}(\lambda^{\infty}_{\text{OL}}(\Pi^*, \epsilon); \Pi^*, \epsilon)} \right\} = 
\sup_{\Pi^*} \left\{ \frac{R^{\infty}(\lambda^{\infty}_{\text{KO}}(q); \Pi^*, \epsilon)}{\inf_{\lambda}R^{\infty}(\lambda; \Pi^*, \epsilon)} \right\}
\end{equation}
is the worst-case (in $\Pi^*$) inflation of the asymptotic risk due to using $\lambda^{\infty}_{\text{KO}}(q)$ instead of the optimal choice for $\lambda$ \citep[our use of the term {\it risk inflation} alludes to the similarity to the criterion proposed in][]{foster1994risk}. 
It is not clear that $\nu(q)$ is at all finite; but if it were true that there exists a choice of $q$ such that for all $\epsilon, \delta, \sigma$, $\nu(q)$ is bounded by $1 + \Delta$ for small $\Delta$, it would be rather reassuring. 

Even for fixed $\epsilon, \delta$ and $\sigma$, computing \eqref{eq:inflation} is not trivial because we need to maximize the ratio over all distributions $\Pi^*$; neither were we able to {\it bound} \eqref{eq:inflation} in general.  
The problem is much simpler in the case where $\Pi^*$ is a point mass at some nonzero (positive) value $\lambda$, because in this case the risk inflation is a univariate function. 
In Figure \ref{fig:point} we plot the risk inflation as a function of $\lambda$ for $\epsilon = 0.1$ and $\delta = 1, \sigma = 0.5$.
The maximum risk inflation is $1.177$ and attained at $t=1.9$, which is quite interesting as this is a case of neither very low nor very high ``signal to noise" ratio. 
Changing $\epsilon$ from $0.1$ to $0.05$ yielded maximum risk inflation of only $1.00032$.

\begin{figure}[H]
\centering
 \includegraphics[width=.65\textwidth]{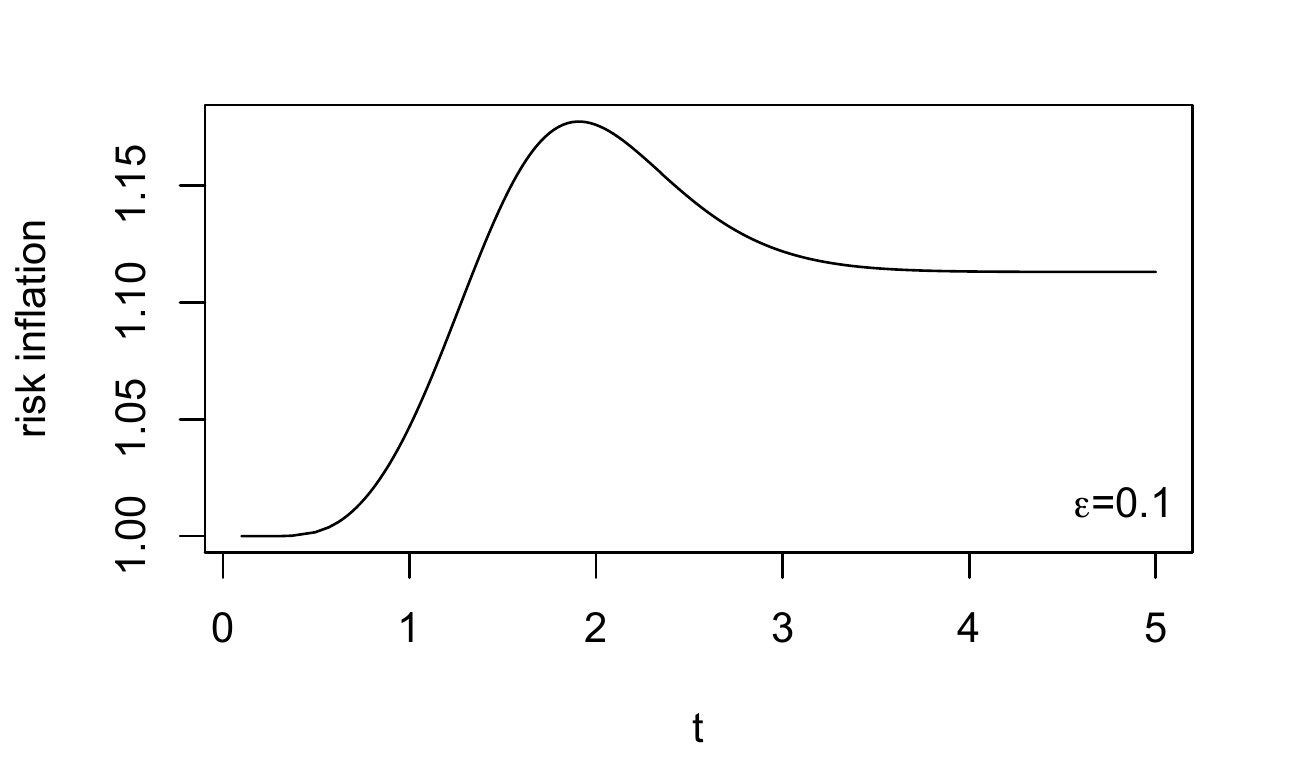}
\caption{
Risk inflation for point mass at $\lambda$, as a function of $\lambda$. 
$\epsilon = 0.1, \ \delta = 1,\ \sigma = 0.5$. 
The maximum $1.177$ is attained at $t=1.9$. 
}
\label{fig:point}
\end{figure}

We continue our investigation with an empirical study. 
Hence, we chose a large number of members of a flexible parametric family of distributions, and evaluate the risk inflation on each. 
Specifically, we took $\delta = 1,\ \sigma = 1$ and $\epsilon = 0.1$, and consider the family of mixtures of Gamma distributions
$$
\left\{  \sum_{i=1}^8 w_i \text{Gam}(\alpha_i): 0\leq w_i\leq 1, \sum_{i=1}^8 w_i = 1  \right\}
$$
where $\boldsymbol{\alpha} = (0.1, 0.8, 1.5, 2.2, 2.9, 3.6, 4.3, 5.0)$ and where $\text{Gam}(\alpha)$ is the Gamma distribution with shape parameter $\alpha$ and rate one. 
We then consider a restriction of the original family, namely, the subset 
$$
\left\{  \sum_{i=1}^8 (u_i/\sum u_j) \text{Gam}(\alpha_i): u_i \in \{0,1,2,3,4\}  \right\}
$$
which includes 383,809 unique members. 
A few of these distributions are plotted in color in Figure \ref{fig:mix} (the black curves correspond to distributions which belong to the original mixture family but not the restricted one).

Next, for each of these 383,809 distributions $\Pi^*$ we compute the ratio
$$
\frac{R^{\infty}(\lambda^{\infty}_{\text{KO}}(q); \Pi^*, \epsilon)}{\inf_{\lambda}R^{\infty}(\lambda; \Pi^*, \epsilon)}
$$
numerically. 
For $q=0.7$, this ratio was always below $1.13$. 
Figure \ref{fig:hist-ratio} shows a histogram of the 383,809 values: 
the distribution with the maximum risk inflation appears in Figure \ref{fig:mix} in red; the ``runner-ups" look similar. 
We also computed the ratio for some distributions that belong to the original mixture family but not the restricted one, and these are shown in black in Figure \ref{fig:mix}; the ratio for both was less than $1.12$. 
To summarize, for $q=0.7$ the risk inflation did not exceed $17\%$ in any of our examples.

\begin{figure}[H]
\centering
 \includegraphics[width=.5\textwidth]{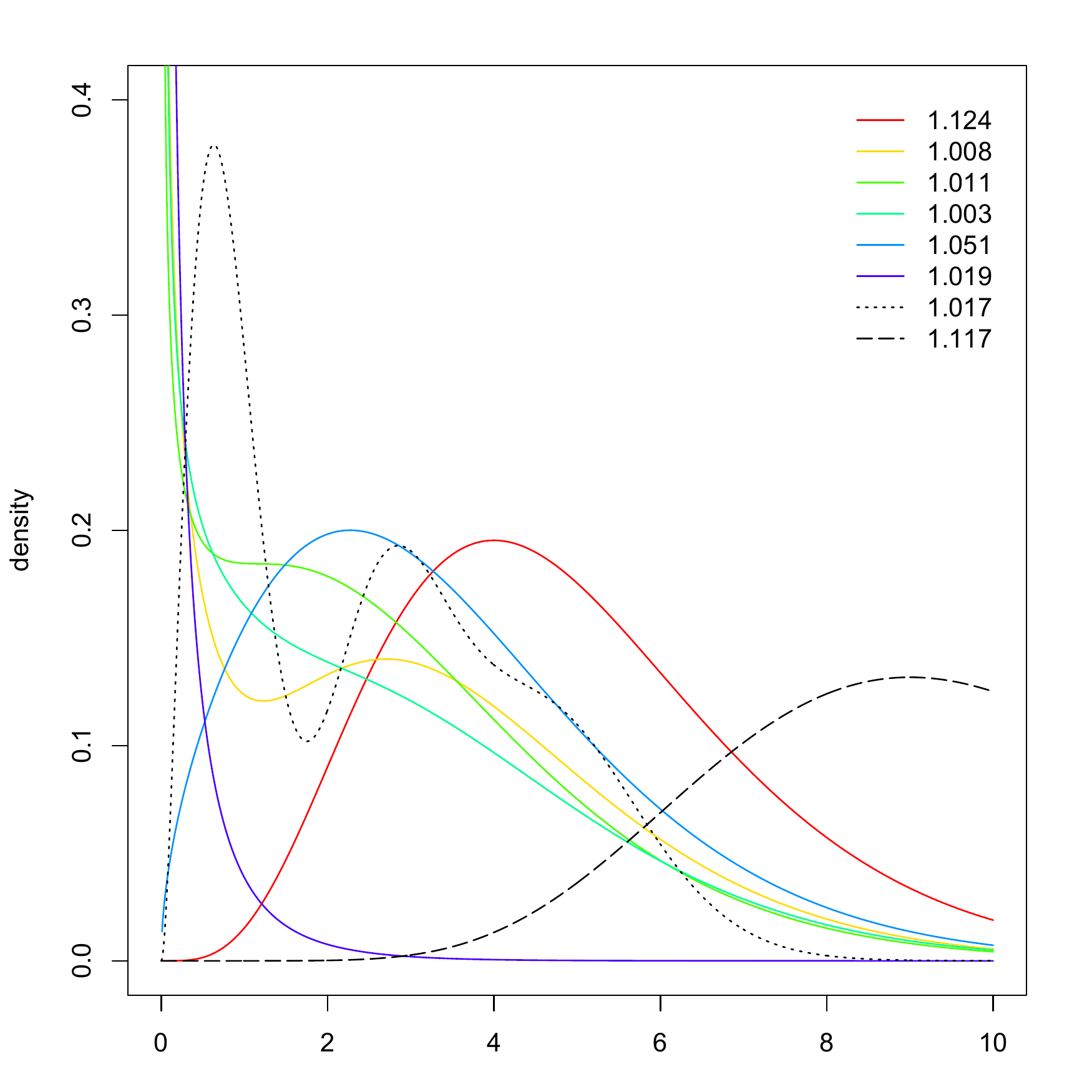}
 \caption{ Gamma mixture distributions.  The colored lines represent
   distributions in the restricted family of 383,809 distributions.
   Black lines are distributions that belong to the original mixture
   family, but not the restricted one.  The red curve corresponds to
   the distribution that attained the maximum ratio among the 383,809
   distributions.  Numbers in legend are risk inflation values
   corresponding to the different distributions in the plot.  }
\label{fig:mix}
\end{figure}

\begin{figure}[H]
\centering
 \includegraphics[width=.65\textwidth]{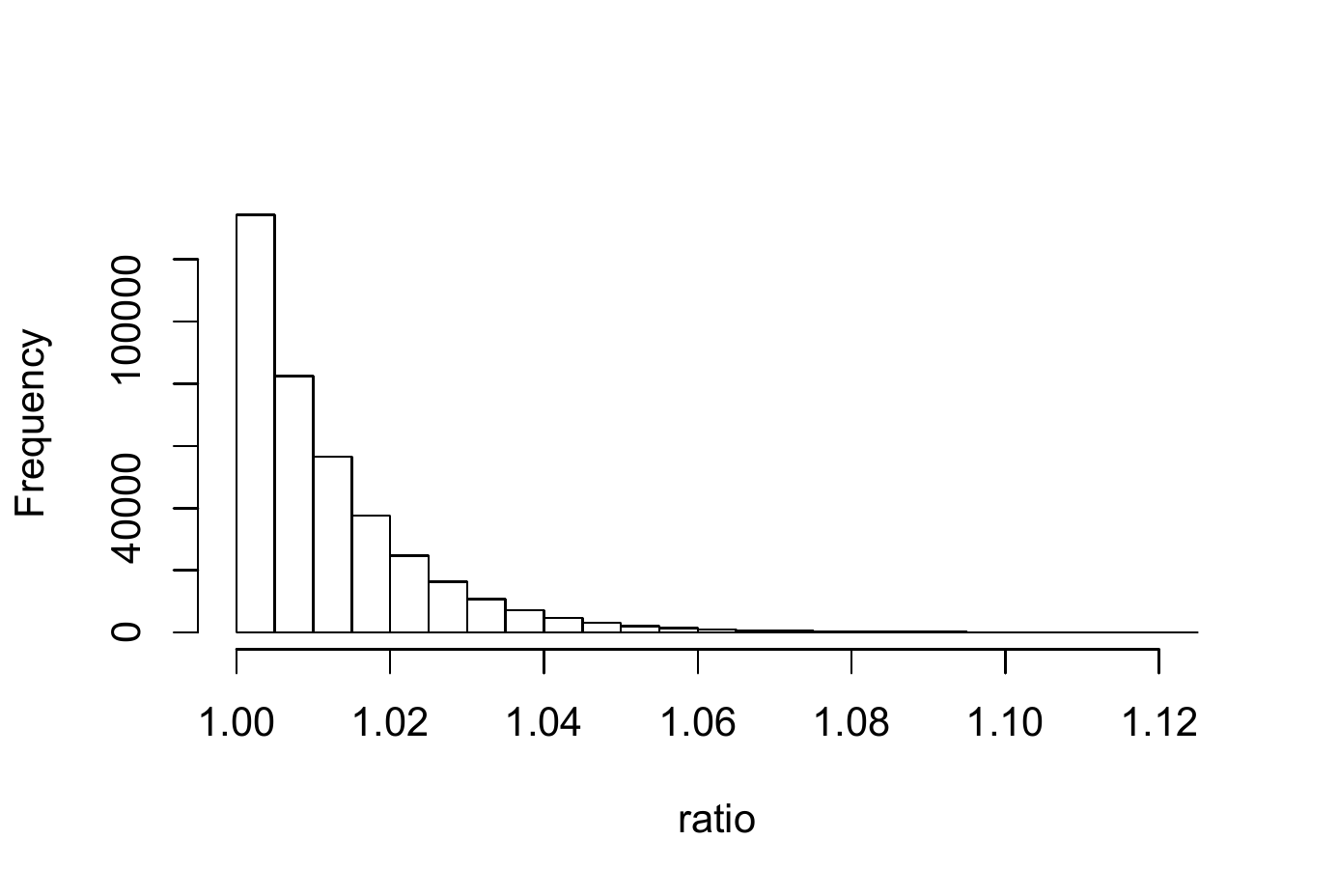}
\caption{
Histogram of the 383,809 computed risk ratios.
}
\label{fig:hist-ratio}
\end{figure}


\section{Discussion}\label{sec:discussion}
The Lasso is nowadays frequently used for selecting variables in a
linear model.  Here we have considered the controlled variable
selection problem, where a valid procedure must satisfy that the
expected proportion of incorrectly selected variables (the FDR) is
kept below a pre-specified level $q$. 
This connects our work to that of \citet{tibshirani2016exact} and 
\citet{fithian2015selective}, who consider hypothesis testing for procedures 
which sequentially select variables into a model. 
However, the null hypothesis tested in these papers is different, corresponding 
to an ``incremental" test versus a ``full-model" test; see the elaborate 
and illuminating discussion in \citet{fithian2015selective}. 


We have proposed an adaptation
of the Lasso-based knockoff procedure of
\citet{barber2015controlling}, which has provable FDR guarantees.
Under the working assumptions, which include a Gaussian i.i.d.~design,
it is possible to analyze the asymptotic tradeoff between the false
discovery proportion and power of the proposed procedure.  Also, while
the i.i.d.~$\calN(0,1)$ design is perhaps not realistic, our aim is to
compare knockoffs against the best possible mechanism for selection
over the lasso path, and such a mechanism has been developed
explicitly only for the i.i.d.~$\calN(0,1)$ setting.  We have
demonstrated that the proposed procedure comes close in performance to
an oracle procedure, which uses knowledge about the underlying signal
to choose the Lasso tuning parameter in such a way that FDP is
controlled.

Moving away from the testing problem, we have also demonstrated that FDR calibration with knockoffs has potential in achieving good prediction performance, adaptively in the sparsity $\epsilon$ and in $\Pi^*$. 

In this article the focus was on procedures tied to the Lasso path. 
There is certainly room for investigating the use of other test statistics to measure feature importance. 
Using the Lasso statistics $T_j$ in \eqref{eq:stat} for the knockoff procedure might be suboptimal, as can be seen in Figure \ref{fig:tradeoff-distributions} and as was implied in \citet{su2017false}: in the Lasso path false discoveries are interspersed between true ones. 
As a matter of fact, an optimal oracle procedure exists for the setup considered in the current paper. 
Namely, the procedure which rejects the hypotheses with $T^*_j \leq \tau$ where
$$
T_j^* = \PP(\beta_j = 0|\by, \bX)
$$
is the ``local" false discovery rate at $(\by,\bX)$ and where $\tau$ is selected so that the procedure has FDR control at $q$. 
\citet{candes2016panning} have already proposed to use similar Bayesian statistics. 
Note that $T_j^*$ defined above depends on $\Pi^*$ and $\epsilon$, hence we associate it with an oracle. 
It would be interesting to think if there is a way to mimic this oracle with a computationally feasible empirical Bayes procedure.

Even if we commit to Lasso statistics, it might be possible to increase power by ordering  the variables differently. 
For example, \citet{candes2016panning} suggested to use the absolute value of the Lasso coefficient, evaluated at the cross-validated $\lambda$, to measure feature importance. 
We leave for future work the investigation of the potential gain in power in using alternative knockoff statistics.

\appendix

\section{Proofs}

\begin{proof}[Proof of Lemma \ref{lem:hyper}]
By definition,
$$
\EE Y = \sum_{1\leq k\leq m} \frac{k}{1+m-k} \cdot \frac{{n_0 \choose k} {n_1\choose m-k}}{{n_0+n_1 \choose m}}
$$
Assuming $n_0, n_1>0$, for $1\leq k\leq m$ we have 
$$
k {n_0\choose k} = n_0{n_0-1 \choose {k-1}},\ \ \ \ \ \ \ \ \ \ \frac{1}{1+m-k}{n_1\choose {m-k}} = \frac{1}{n_1 + 1} {n_1+1 \choose {m-k+1}}.
$$
Therefore, 
$$
\EE Y = \frac{n_0}{1+n_1}\cdot \frac{1}{ {n_0+n_1\choose m} } \sum_{1\leq k \leq m} {n_0-1 \choose {k-1}} \cdot {n_1+1 \choose {m-k+1}} = \frac{n_0}{1+n_1} \left( 1 - \frac{ {n_0-1 \choose m} }{ {n_0+n_1\choose m} } \right). 
$$
If $n_0 = 0$, then $X=Y=0$. 
If $n_1 = 0$, then $X=Y=m$. 
In both cases the claimed result still holds. 
\end{proof}

\bigskip

\begin{proof}[Proof of identity \eqref{eq:hypergeometric-identity}]
Assuming $m_0,r>0$, we have 
\begin{align*}
&\EE \left( \frac{X}{1+m-X} \cdot \frac{r+X-m}{1+m_0-X} \right) = \sum_{ k = 1\vee (m-r+1) }^{m\wedge m_0} \frac{k}{1+m-k} \cdot \frac{r+k-m}{1+m_0-k} \cdot \frac{{m_0\choose k}{r\choose {m-k}}}{{m_0+r\choose m}}\\
&=\sum_{ k = 1\vee (m-r+1) }^{m\wedge m_0} \frac{m_0!}{(k-1)!(m_0-k+1)!} \cdot \frac{r!}{(m-k+1)!(r-m+k-1)!} \cdot \frac{m!(m_0+r-m)!}{(m_0+r)!}\\
&=\sum_{ k = 1\vee (m-r+1) }^{m\wedge m_0} \frac{ {m_0\choose {k-1}} {r\choose m-k+1} }{ {m_0+r\choose m} } = 1 - \frac{ {m_0\choose m_0\wedge m} {r\choose m-m_0\wedge m} }{ {m_0+r\choose m} } 
\end{align*}

If $m_0 = 0$, then $X=0$ and the identity still holds. 
If $r = 0$, then $X=m$ and the identity again continues to holds. 
\end{proof}

\section{Computing the curve $q^{\Pi^*}$}

To compute $q^{\Pi^*}(\tppinfty(\lambda); \epsilon, \delta, \sigma)$ for specific $\Pi^*$ and $\epsilon$, we need to be able to compute the pair $(\fdpinfty(\lambda),\tppinfty(\lambda))$ for each $\lambda>0$. 
This, in turn, requires to find for each $\lambda$ the pair $(\alpha, \tau)$ given by the system of equations \eqref{eq:system}. 
Hence, the functionals  
$$
f(\Pi^*, \epsilon) := \EE(\eta_{\alpha\tau}(\Pi + \tau W) - \Pi)^2 = (1-\epsilon)\EE(\eta_{\alpha\tau}(\tau W) )^2 + \epsilon \EE(\eta_{\alpha\tau}(\Pi^* + \tau W) - \Pi^*)^2
$$
and 
$$
g(\Pi^*, \epsilon) := \PP(|\Pi + \tau W|>\alpha\tau) = 2(1-\epsilon)\Phi(-\alpha) + \epsilon \PP(|\Pi^* + \tau W|>\alpha\tau),
$$ 
which appear in \eqref{eq:system}, are evaluated numerically for a given CDF $F_{\Pi^*}$ of $\Pi^*$. 

For a fixed value of $\lambda$, our code takes as an input $\epsilon$ and a CDF $F_{\Pi}$ of $\Pi$ (as well as $\delta$ and $\sigma^2$) and returns a solution for $(\alpha, \tau)$ in \eqref{eq:system}. 
Thus, for given values of $\alpha$ and $\tau$, we need to be able to compute 
$
f(\Pi, \epsilon) 
$
and 
$
g(\Pi, \epsilon) 
$ 
defined in Section \ref{sec:tradeoff}, 
as functionals of $F_{\Pi}$. 
We explain how we implement this. 

Let $R(\mu):= \EE \left(\eta_{\alpha\tau}(\mu + \tau W) - \mu\right)^2$ and $Q(\mu):=\PP\left( |\mu + \tau W|>\alpha\tau \right)$. 
Writing $f(\Pi, \epsilon)$ as a Riemann-Stieltjes integral and using integration by parts, we have
\begin{align*}
f(\Pi, \epsilon) &:= \EE(\eta_{\alpha\tau}(\Pi + \tau W) - \Pi)^2 = \int_{-\infty}^\infty R(\mu) \ dF_{\Pi} \\
&= R(\infty)F_{\Pi}(\infty) - R(-\infty)F_{\Pi}(-\infty) - \int_{-\infty}^\infty F_{\Pi}(\mu) \ dR(\mu)\\
&= R(\infty)F_{\Pi}(\infty) - R(-\infty)F_{\Pi}(-\infty) - \int_{-\infty}^\infty F_{\Pi}(\mu) R'(\mu) d\mu\\
&= 1 + \alpha^2\tau^2 - \int_{-\infty}^\infty F_{\Pi}(\mu) [ 2\mu \left( \Phi\left( \alpha - \mu/\tau \right) \right) - \left( \Phi\left( -\alpha - \mu/\tau \right) \right) ] d\mu
\end{align*}
where we substituted a closed-form expression for the derivative of the risk of soft thresholding \citep[see, e.g., ][Ch. 2.7]{johnstone2011gaussian}. 
The last expression is evaluated using numerical integration. 

Similarly, 
\begin{align*}
g(\Pi, \epsilon) &:= \PP(|\Pi + \tau W|>\alpha\tau) = \int_{-\infty}^\infty Q(\mu) \ dF_{\Pi} \\
&= Q(\infty)F_{\Pi}(\infty) - Q(-\infty)F_{\Pi}(-\infty) - \int_{-\infty}^\infty F_{\Pi}(\mu) \ dQ(\mu)\\
&= Q(\infty)F_{\Pi}(\infty) - Q(-\infty)F_{\Pi}(-\infty) - \int_{-\infty}^\infty F_{\Pi}(\mu) Q'(\mu) \ d\mu\\
&= 1- 1/\tau \int_{-\infty}^\infty F_{\Pi}(\mu) [ \phi(\alpha - \mu/\tau) - \phi(-\alpha - \mu/\tau) ] \ d\mu, 
\end{align*}
and, again, the last expression is evaluated using numerical integration.

\bibliographystyle{plainnat}
\bibliography{/Users/assafweinstein/Dropbox/Research/References.bib}

\end{document}